\documentclass[a4paper,draft]{article}
\pdfoutput=1

\usepackage{url}
\usepackage[hidelinks]{hyperref}
\usepackage[utf8]{inputenc}
\usepackage[small]{caption}
\usepackage{graphicx}
\usepackage{amsmath}
\usepackage{amsthm}
\usepackage{booktabs}
\usepackage{algorithm}
\usepackage{algorithmic}
\usepackage{breakcites}
\usepackage{natbib}
\newtheorem{theorem}{Theorem}
\newtheorem{lemma}[theorem]{Lemma}
\newtheorem{proposition}[theorem]{Proposition}
\newtheorem{definition}{Definition}

\usepackage{relsize} %
\usepackage{esvect} %
\usepackage{amssymb} 
\usepackage{todonotes}
\usepackage{xcolor}
\usepackage{environ}
\usepackage{tabularx}

\newcommand{\powerset}[1]{2^{#1}}
\newcommand{\naturalNumbers}{\mathbb{N}}
\newcommand{\setOfSatisfactionDegrees}{\overline{\naturalNumbers}}
\newcommand{\replace}[2]{[{#1}/{#2}]}

\newcommand{\sizeof}[1]{|{#1}|}
\newcommand{\bigO}{\mathcal{O}}
\renewcommand{\max}{\mathit{max}}
\renewcommand{\min}{\mathit{min}}

\newcommand{\pl}{\textnormal{PL}} %
\newcommand{\qcl}{\textnormal{QCL}} 
\newcommand{\pqcl}{\textnormal{PQCL}}
\newcommand{\qclPlus}{\textnormal{QCL+}}  
\newcommand{\ccl}{\textnormal{CCL}} 
\newcommand{\qccl}{\textnormal{QCCL}} %

\newcommand{\sccl}{\textnormal{SCCL}}
\newcommand{\lcl}{\textnormal{LCL}}

\newcommand{\cl}{\mathcal{L}}

\newcommand{\universe}{\mathcal{U}}
\newcommand{\connectives}[1]{\mathcal{C}_{#1}}
\newcommand{\formulas}[1]{\mathcal{F}_{#1}}
\newcommand{\variablesOf}[1]{var({#1})}
\newcommand{\generalConn}{\circ}
\newcommand{\qclConn}{\vv{\times}}
\newcommand{\cclConn}{\vv{\odot}}

\newcommand{\scclConn}{\vv{\circledast}}
\newcommand{\lclConn}{\vv{\diamond}}

\newcommand{\I}{\mathcal{I}}
\newcommand{\J}{\mathcal{J}}

\newcommand{\opt}[1]{\mathit{opt}_{#1}}
\newcommand{\degree}[1]{\mathit{deg}_{#1}}
\newcommand{\optConn}[2]{\opt{{#1}}^{#2}}
\newcommand{\degreeConn}[2]{\degree{{#1}}^{#2}}
\newcommand{\sat}[2]{\models^{#1}_{#2}}

\newcommand{\prefModels}[2]{\mathit{Prf}_{#1}({#2})}
\newcommand{\obtainableDegrees}[1]{\mathcal{D}_{#1}}

\newcommand{\degreeEquiv}[1]{\equiv^{#1}_{d}}
\newcommand{\fullEquiv}[1]{\equiv^{#1}_{f}}
\newcommand{\strongEquiv}[1]{\equiv^{#1}_{s}}

\newcommand{\ccfont}[1]{\protect\mathsf{#1}}
\newcommand{\PolyTime}{\ccfont{P}}
\newcommand{\NP}{\ccfont{NP}}
\newcommand{\coNP}{\ccfont{coNP}}

\newcommand{\phDelta}[1]{\Delta_{#1}\PolyTime}
\newcommand{\phTheta}[1]{\Theta_{#1}\PolyTime}

\newcommand{\SAT}{\textsc{Sat}}

\newcommand{\Unsat}{\textsc{Unsat}}
\newcommand{\LexMaxSat}{\textsc{LexMaxSat}}
\newcommand{\LogLexMaxSat}{\textsc{LogLexMaxSat}}
\newcommand{\clModelChecking}[1]{{#1}-\textsc{DegreeChecking}}
\newcommand{\clSAT}[1]{{#1}-\textsc{DegreeSat}}
\newcommand{\clPrefModelChecking}[1]{{#1}-\textsc{PrefModelChecking}}
\newcommand{\clPrefModelSAT}[1]{{#1}-\textsc{PrefModelSat}}
\newcommand{\clDegreeEquivChecking}[1]{{#1}-\textsc{DegreeEquivalence}}
\newcommand{\clFullEquivChecking}[1]{{#1}-\textsc{FullEquivalence}}
\newcommand{\clStrongEquivChecking}[1]{{#1}-\textsc{StrongEquivalence}}

\makeatletter
\newcommand{\problemtitle}[1]{\gdef\@problemtitle{#1}}%
\newcommand{\probleminput}[1]{\gdef\@probleminput{#1}}%
\newcommand{\problemquestion}[1]{\gdef\@problemquestion{#1}}%
\NewEnviron{problem}{
	\problemtitle{}\probleminput{}\problemquestion{}%
	\BODY%
	\par\addvspace{.4\baselineskip}
	\noindent
	\begin{tabular}{@{\hspace{0pt}}lp{0.752\columnwidth}}
		\multicolumn{2}{@{\hspace{0pt}}l}{\textsc{\@problemtitle}} \\
		Input: & \@probleminput \\
		Question: & \@problemquestion%
	\end{tabular}
	\par\addvspace{.4\baselineskip}
}

\title{Choice Logics and Their \\ Computational Properties}

\author{
	Michael Bernreiter,
	Jan Maly,
	Stefan Woltran
	\\
	Institute of Logic and Computation, TU Wien, Austria
	\\
	\{mbernrei, jmaly, woltran\}@dbai.tuwien.ac.at
}

\date{}

\begin{document}
	
\maketitle

\begin{abstract}
	Qualitative Choice Logic ($\qcl$) and Conjunctive Choice Logic ($\ccl$) are formalisms for preference handling, with especially $\qcl$ being well established in the field of AI. 
	So far, analyses of these logics need to be done on a case-by-case basis, albeit 
	they share several common features. 
	This calls for a more general choice logic framework, 
	with $\qcl$ and $\ccl$ as well as some of their derivatives being
	particular instantiations. We provide such a framework, which allows us, on the one hand,
	to easily define new choice logics and, on the other hand,
	to examine properties of different choice logics in a uniform setting. In particular, 
	we investigate strong equivalence, a core concept in non-classical logics for
	understanding formula simplification, and computational complexity. 
	Our analysis %
	also yields new results for $\qcl$ and $\ccl$. For example, we show that the main reasoning task regarding preferred models is $\phTheta{2}$-complete for $\qcl$ and $\ccl$, while being $\phDelta{2}$-complete for a newly introduced choice logic.
\end{abstract}

\section{Introduction}

Representing preferences and reasoning about them 
is a key challenge in many areas of 
AI research. One of the most fruitful approaches to preference representation
has been the use of logic-based formalisms \citep{domshlak2011preferences,pigozzi2016preferences}.
Two closely related examples from the literature are Qualitative Choice Logic ($\qcl$)
\citep{brewka2004qualitative} and Conjunctive Choice Logic ($\ccl$)
\citep{boudjelida2016conjunctive}. 
Especially $\qcl$ has proven to be a useful preference formalism,
with applications ranging from logic programming \citep{brewka2004logic}
to alert correlation \citep{benferhat2008alert} to database querying \citep{lietard2014towards}. 
However, several key computational properties of $\qcl$ and $\ccl$ have not been studied yet.
This includes strong equivalence, a tool to understand formula simplification, and the computational complexity of main reasoning tasks.

Moreover, the two types of preferences expressed by
$\qcl$ and $\ccl$ are certainly not the only ones. 
$\qcl$ extends classical propositional logic with a non-classical connective $\qclConn$ called ordered disjunction. Intuitively, $F \qclConn G$ means that it is preferable to satisfy $F$ but, if that is not possible, then satisfying $G$ is also acceptable.
Similarly, $\ccl$ introduces ordered conjunction
($\cclConn$), where the intended meaning of $F \cclConn G$ is that it is most preferable to satisfy
both $F$~and~$G$, but satisfying only $F$ is also acceptable. More specifically,
interpretations ascribe a number, called satisfaction degree, to $\qcl$- and
$\ccl$-formulas. The preferred models of a formula are those models with the least degree.
Other natural preferences cannot be succinctly represented in $\qcl$ and $\ccl$. For example, one could think of a more fine-grained choice connective: given $F \generalConn G$, it would be best to satisfy both $F$~and~$G$, second best to satisfy only $F$, and third best to satisfy only $G$. Or one could desire a connective with the same basic behavior as $\qclConn$ in $\qcl$, but in which satisfaction degrees are handled in a different way. There is a multitude of interesting logics related to $\qcl$ and $\ccl$ that have yet to be defined, and which may very well prove to be as useful as $\qcl$. 

In this paper, we propose a general framework for choice logics that,
on the one hand, makes it easy to define new choice logics by specifying 
one or more choice connectives and, on the other hand, allows us to settle
open questions regarding the computational properties of $\qcl$ and $\ccl$ in a 
uniform way. 
In detail, our main contributions are as follows:

\begin{itemize}
	\item We formally define a framework that captures both $\qcl$ and $\ccl$, as well as infinitely many new related logics. 
	To showcase the versatility of our framework we explicitly introduce two such new logics called  Lexicographic Choice Logic ($\lcl$) and Simple Conjunctive Choice Logic ($\sccl$).
	\item We characterize strong equivalence via simpler equivalence notions for large classes of choice logics. This further enables us to analyze properties related to strong equivalence more easily, and also provides valuable insights into the nature of choice logics.
	\item We analyze the computational complexity of choice logics in detail.\footnote{
		The complexity of some decision problems pertaining to $\qcl$ was conjectured by \citet{lang2004logical} but never formally investigated.} 
	For example, we show that the complexity of the main decision problem regarding preferred models ranges from $\NP$- to $\phDelta{2}$-completeness with QCL and CCL being located in between ($\phTheta{2}$-complete). 
	The complexity of checking for strong equivalence follows from our characterization via simpler equivalence notions.
\end{itemize}

\paragraph{Related Work}
$\qcl$ and $\ccl$ are not the only logic-based preference representation 
formalisms.
Other prominent examples include the preference logics introduced by
\citet{wright1963logicofpreference} 
and \citet{benthem2009everythingelse};
for a more complete overview, see the surveys by 
\citet{domshlak2011preferences} and \citet{bienvenu2010preference}. Most of these %
formalisms differ from choice 
logics in that they only represent preferences, while $\qcl$ and $\ccl$ integrate the 
representation of truth and preference
in one formalism.
One exception is recent work by \citet{arxivLexicographicLogic}, which is conceptually closely related to our $\lcl$, but differs in that formulas are assigned lists of truth values instead of satisfaction degrees. 

From a technical standpoint, many non-monotonic logics are closely related to $\qcl$, as they are inherently connected to preferences \citep{DBLP:conf/ijcai/Shoham87}. This is particularly true for propositional circumscription and possibilistic logic \citep{brewka2004qualitative}.
However, unlike choice logics, these formalisms are not primarily designed to represent preferences and often rely on constructs outside of the logical language (circumscription policy, possibility distribution) to represent knowledge.
It is also worth noting that choice logics are technically very different
from traditional infinite-valued logics \citep{gottwald2001treatise} as, for example, choice logics use classical interpretations that make atoms either true or false.

There are also systems that are based on $\qcl$'s ordered disjunction, but do not fit into our framework. 
For instance, \citet{jiang2015logic}
introduce a modal logic that contains a binary connective with a similar meaning to 
ordered disjunction, while 
\citet{zhang2015representing} state that they took inspiration from $\qcl$ for their prioritized disjunction, which is used to reason about game strategies. 

The remainder of the paper is structured as follows: Section~\ref{sec:framework} contains the definition of our framework, examples for logics that can be defined within it, and a result on the expressiveness of choice logics. The notion of strong equivalence is examined in Section~\ref{sec:strongEquivalence}, and the complexity of choice logics is analyzed in Section~\ref{sec:complexity}. Finally, Section~\ref{sec:conclusion} contains a summary of our results and pointers to future work.

\section{Choice Logic Framework} \label{sec:framework}

We shall start with introducing our general framework, then show how both
existing and new logics can be defined in this framework, and finally give a 
synthesis result that holds for any choice logic of our framework.

In what follows, 
$\pl$ stands for classical propositional logic,
$\universe$ denotes the alphabet of propositional variables, 
and an interpretation $\I$ is defined as a set of propositional variables such that $a \in \I$ if and only if $a$ is set to true by $\I$. If $\I$ satisfies a classical formula $F$, we write $\I \models F$.

\subsection{Syntax and Semantics} \label{sec:syntaxAndSemantics}

A choice logic has two types of connectives: classical connectives (here we use $\neg$, $\land$, and $\lor$), and binary choice connectives, with which preferences can be expressed.%

\begin{definition}\label{def:clFormulas}
	The set of choice connectives $\connectives{\cl}$ of a choice logic $\cl$ is a finite set of symbols such that $\connectives{\cl} \cap \{\neg, \land, \lor\} = \emptyset$. The set $\formulas{\cl}$ of formulas of %
	$\cl$ is defined inductively as follows:
	\begin{enumerate}
		\item $a \in \formulas{\cl}$ for all $a \in \universe$;
		\item if $F \in \formulas{\cl}$, then $(\neg F) \in \formulas{\cl}$;
		\item if $F, G \in \formulas{\cl}$, then $(F \generalConn G) \in \formulas{\cl}$ for $\generalConn \in (\{\land, \lor\} \cup \connectives{\cl})$. 
	\end{enumerate}
	$\variablesOf{F}$ denotes the set of all variables in a formula $F\in\formulas{\cl}$. 
\end{definition}

For example, $\connectives{\qcl} = \{\qclConn\}$ and $\connectives{\pl} = \emptyset$. Formulas that do not contain a choice connective are simply classical formulas. 

The semantics of a choice logic is given by two functions, satisfaction degree and optionality.
The satisfaction degree of a formula given an interpretation is either a natural number or $\infty$.
The lower this degree, the more preferable the interpretation. The optionality of a formula
describes the maximum finite satisfaction degree that this formula can be ascribed.
As we will see in Section~\ref{sec:examples}, optionality is used to penalize less preferable
interpretations.

\begin{definition} \label{def:clOptionality}
	The optionality of a choice connective $\generalConn \in \connectives{\cl}$ in a choice logic $\cl$ is given by a function $\optConn{\cl}{\generalConn} \colon \naturalNumbers^2 \to \naturalNumbers$ such that $\optConn{\cl}{\generalConn}(k,\ell) \leq (k+1)\cdot(\ell+1)$ for all $k,\ell \in \naturalNumbers$. The optionality of an $\cl$-formula is given via %
	$\opt{\cl} \colon \formulas{\cl} \to \naturalNumbers$ with %
	\begin{enumerate}
		\item $\opt{\cl}(a) = 1$, for every $a \in \universe$;
		\item $\opt{\cl}(\neg F) = 1$;
		\item $\opt{\cl}(F \land G) = \max(\opt{\cl}(F),\opt{\cl}(G))$;
		\item $\opt{\cl}(F \lor G) = \max(\opt{\cl}(F),\opt{\cl}(G))$; 
		\item $\opt{\cl}(F \generalConn G) = \optConn{\cl}{\generalConn}(\opt{\cl}(F),\opt{\cl}(G))$ for every
		$\generalConn \in \connectives{\cl}$.
	\end{enumerate}
\end{definition}

The optionality of a classical formula is always $1$. For any choice connective~$\generalConn$, the optionality of $F \generalConn G$ is bounded such that $\opt{\cl}(F \generalConn G) \leq (\opt{\cl}(F)+1) \cdot (\opt{\cl}(G)+1)$.
The reason for this is that there are $\opt{\cl}(F)$ many finite degrees that could be ascribed to $F$, plus the infinite degree $\infty$. Likewise for $G$. Thus, there are at most $(\opt{\cl}(F)+1) \cdot (\opt{\cl}(G)+1)$ possibilities when combining the degrees of $F$ and $G$.

Next, we define the satisfaction degree of a formula. In the following, we write $\setOfSatisfactionDegrees$ for $(\naturalNumbers \cup \{\infty\})$. 

\begin{definition}\label{def:clSatisfactionDegree}
	The satisfaction degree of a choice connective $\generalConn \in \connectives{\cl}$ in a choice logic $\cl$ is given by a function ${\degreeConn{\cl}{\generalConn} \colon \naturalNumbers^{2} \times {\setOfSatisfactionDegrees}^{2} \to \setOfSatisfactionDegrees}$ where either $\degreeConn{\cl}{\generalConn}(k,\ell,m,n) \leq \optConn{\cl}{\generalConn}(k,\ell)$ or $\degreeConn{\cl}{\generalConn}(k,\ell,m,n) = \infty$ holds for all $k,\ell \in \naturalNumbers$ and all $m,n \in \setOfSatisfactionDegrees$. The satisfaction degree of an $\cl$-formula under an interpretation is given via $\degree{\cl} \colon \powerset{\universe} \times \formulas{\cl} \to \setOfSatisfactionDegrees$ with
	\begin{enumerate}
		\item $\degree{\cl}(\I,a) = 
		\begin{cases}
		1 & \text{if } a \in \I, \\
		\infty & \text{otherwise},
		\end{cases}$ \hspace{0.2cm} 
		for every $a \in \universe$;
		\item $\degree{\cl}(\I,\neg F) = 
		\begin{cases}
		1 & \text{if }  \degree{\cl}(\I,F) = \infty, \\
		\infty & \text{otherwise};
		\end{cases}$
		\item $\degree{\cl}(\I,F \land G) = \max(\degree{\cl}(\I,F),\degree{\cl}(\I,G))$;
		\item $\degree{\cl}(\I,F \lor G) = \min(\degree{\cl}(\I,F),\degree{\cl}(\I,G))$;
		\item $\degree{\cl}(\I,F \generalConn G) =  \degreeConn{\cl}{\generalConn}(\opt{\cl}(F), \opt{\cl}(G), \degree{\cl}(\I,F), \degree{\cl}(\I,G))$ for every $\generalConn \in \connectives{\cl}$.
	\end{enumerate}
\end{definition}

We also write $\I \sat{\cl}{m} F$ for $\degree{\cl}(\I,F) = m$. If $m < \infty$, we say that $\I$ satisfies $F$ (to a finite degree), and if $m = \infty$, then $\I$ does not satisfy $F$. If $F$ is a classical formula, then ${\I \sat{\cl}{1} F \iff \I \models F}$ and ${\I \sat{\cl}{\infty} F \iff \I \not\models F}$. The symbols $\top$ and $\bot$ are shorthand for the formulas $(a \lor \neg a)$ and $(a \land \neg a)$, where $a$ can be any variable. %
We have $\opt{\cl}(\top) = \opt{\cl}(\bot) = 1$, $\degree{\cl}(\I,\top) = 1$ and $\degree{\cl}(\I,\bot) = \infty$ %
for any interpretation $\I$ in every choice logic.

The semantics of the classical connectives are fixed and are the same as for QCL and CCL.
$F \land G$ is assigned the maximum degree of $F$ and $G$ because both formulas need to be satisfied. Conversely, we use the minimum degree for $F \lor G$ since satisfying either option suffices, and we do not need to concern ourselves with the less preferable option. Observe that it is still necessary to define $\opt{\cl}(F \lor G) = \max(\opt{\cl}(F),\opt{\cl}(G))$, as the case that either option is not satisfied has to be allowed for. As for negation, note that $\neg F$ can only assume the degrees $1$ or $\infty$.
Therefore, $\neg$ can be seen as classical negation when applied to a classical formula, and
as a tool to neutralize satisfaction degrees.
In order to define a form of negation that results in different degrees of satisfaction, 
we would need to keep track of degrees of dissatisfaction.
We observe that the semantics of the classical connectives used here are not the only possible ones. 
See also the brief discussion on $\pqcl$ and $\qclPlus$ in Section~\ref{sec:examples}.

From Definitions~\ref{def:clOptionality} and~\ref{def:clSatisfactionDegree} it follows that the satisfaction degree of a choice logic formula is bounded by its optionality as intended: 
\begin{lemma} \label{lemma:degreeIsBoundedByOptionality}
	Let $\cl$ be a choice logic. Then for all interpretations $\I$ and all $\cl$-formulas $F$, either $\degree{\cl}(\I,F) \leq \opt{\cl}(F)$ or $\degree{\cl}(\I,F) = \infty$.
\end{lemma}
\begin{proof}
	By structural induction. Let $\cl$ be a choice logic and $\I$ an interpretation. 
	
	Base case: if $a \in \universe$, then either $\degree{\cl}(\I, a) = 1 = \opt{\cl}(a)$ or $\degree{\cl}(\I, a) = \infty$.
	
	Step case: let $F$ and $G$ be $\cl$-formulas with $\opt{\cl}(F) = k$, $\opt{\cl}(G) = \ell$, $\degree{\cl}(\I,F) = m$, and $\degree{\cl}(\I,G) = n$. By the I.H., either $m \leq k$ or $m = \infty$. Likewise, $n \leq \ell$ or $n = \infty$. The case for $\neg F$ is analogous to the base case. For $F \land G$, either $\degree{\cl}(\I, F \land G) = \max(m,n) \leq max(k,\ell) = \opt{\cl}(F \land G)$ or $\degree{\cl}(\I, F \land G) = \infty$. Analogous for $F \lor G$. For any choice connective $\generalConn \in \connectives{\cl}$, we have that ${\opt{\cl}(F \generalConn G) = \optConn{\cl}{\generalConn}(k, \ell)}$ and $\degree{\cl}(\I,F \generalConn G) = \degreeConn{\cl}{\generalConn}(k, \ell, m, n)$. By Definition~\ref{def:clSatisfactionDegree}, either $\degreeConn{\cl}{\generalConn}(k, \ell, m, n) \leq \optConn{\cl}{\generalConn}(k, \ell)$ or $\degreeConn{\cl}{\generalConn}(k, \ell, m, n) = \infty$. 
\end{proof}

Moreover, note that only those variables that actually occur in a formula $F$ can influence $\opt{\cl}(F)$ and $\degree{\cl}(\I,F)$, which means that we can assume $\I \subseteq \variablesOf{F}$ for any interpretation $\I$ that we are dealing with:

\begin{lemma} \label{lemma:nonOccuringAtomsHaveNoInfluence}
	Let $\cl$ be a choice logic, $\I$ an interpretation, and $F$ an $\cl$-formula. Let $\J = \I \cap \variablesOf{F}$. Then $\degree{\cl}(\I,F) = \degree{\cl}(\J,F)$.
\end{lemma}
\begin{proof}
	This follows directly from the fact that the semantics of all connectives in a choice logic are given by functions over the optionalities and satisfaction degrees of their operands. %
\end{proof}

An interpretation that satisfies a formula to a finite degree will be referred to as a \emph{model} of that formula. But often we are interested only in the most preferable models.

\begin{definition}
	Let $\cl$ be a choice logic. Then an interpretation $\I$ is a preferred model of an $\cl$-formula $F$, written as $\I \in \prefModels{\cl}{F}$, if $\degree{\cl}(\I,F) \neq \infty$ and, for all other interpretations $\J$, $\degree{\cl}(\I,F) \leq \degree{\cl}(\J,F)$.
\end{definition}

By the above definition, $\prefModels{\cl}{F} = \emptyset$ if and only if $\degree{\cl}(\I,F) = \infty$ for all interpretations $\I$.

\subsection{Examples of Choice Logics} \label{sec:examples}

To define a choice logic $\cl$, it suffices to specify the choice connectives $\connectives{\cl}$ of that logic, and to provide the optionality- and satisfaction degree functions of every ${\generalConn \in \connectives{\cl}}$. The simplest example for a choice logic is classical propositional logic, where $\connectives{\pl} = \emptyset$. Since all $\pl$-formulas are classical formulas, ${\I \in \prefModels{\pl}{F} \iff \I \models F}$. 

$\qcl$ can be expressed in our framework as follows: 

\begin{definition} \label{def:qcl}
	$\qcl$ is the choice logic such that $\connectives{\qcl} = \{\qclConn\}$, and, if $k = \opt{\qcl}(F)$, $\ell = \opt{\qcl}(G)$, $m = \degree{\qcl}(\I,F)$, and $n = \degree{\qcl}(\I,G)$ for some $F$, $G$, $\I$, then
	\begin{align*}
	\opt{\qcl}(F \qclConn G) & = \optConn{\qcl}{\qclConn}(k,\ell) = k + \ell, \text{ and} \\
	\degree{\qcl}(\I,F \qclConn G) & = \degreeConn{\qcl}{\qclConn}(k,\ell,m,n)  \\ & = 
	\begin{cases}
	m & \text{if } m < \infty; \\
	n + k & \text{if } m = \infty, n <  \infty; \\
	\infty & \text{otherwise}.
	\end{cases}
	\end{align*}
\end{definition}

Here, we can see how optionality is used to penalize non-satisfaction: if $F$ is satisfied, then $\degree{\qcl}(\I,F \qclConn G) = \degree{\qcl}(\I,F) \leq \opt{\cl}(F)$; %
if $F$ is not satisfied, but $G$ is, then $\degree{\qcl}(\I,F \qclConn G) = \degree{\qcl}(\I,G) + \opt{\cl}(F) > \opt{\cl}(F)$. Thus, any interpretation which satisfies $F$ is automatically more preferable than one that does not. 

The fourth column of Table~\ref{table:comparisonOfChoiceConnectives}
illustrates $\qcl$ semantics on a simple example. From there, we can infer that $\{a\}$ and $\{a,b\}$ are preferred models of $a \qclConn b$, while $\emptyset$ and $\{b\}$ are not.

We know that $\qclConn$ is associative from \citet{brewka2004qualitative}, which means that given arbitrary $\qcl$-formulas $A$, $B$, and $C$, the formulas $((A \qclConn B) \qclConn C)$ and $(A \qclConn (B \qclConn C))$ always have the same optionality and satisfaction degrees. We can therefore write $F_1 \qclConn F_2 \ldots \qclConn F_n$ to express that we prefer $F_1$ to $F_2$, $F_2$ to $F_3$, and so on. For %
variables $a_1,\ldots,a_n$ with $a_i \neq a_j$ for all $i \neq j$, we have that $\{a_i\} \sat{\qcl}{i} (a_1 \qclConn a_2 \ldots \qclConn a_n)$.

Let us briefly consider how $\qclConn$ interacts with the classical connectives, specifically conjunction: let $F = (a \qclConn b) \land (c \qclConn d)$. Then all preferred models of $F$ must contain both $a$ and $c$, since otherwise either $(a \qclConn b)$ or $(c \qclConn d)$, and therefore also $F$, are ascribed a degree greater than $1$.

Next, we formally define $\ccl$. However, our function for the satisfaction degree of $\cclConn$ differs from the one given by \citet{boudjelida2016conjunctive}. This is because the original definition of $\ccl$, although it can be expressed in our framework if desired, fails to capture the intended meaning of ordered conjunction.\footnote{In our understanding, under the semantics described by \citet[Definition~8]{boudjelida2016conjunctive}, the formula $a \cclConn b$ will always be ascribed a degree of either $\infty$ or $1$.}

\begin{definition} \label{def:cclAlternative}
	$\ccl$ is the choice logic such that $\connectives{\ccl} = \{\cclConn\}$, and, if $k = \opt{\ccl}(F)$, $\ell = \opt{\ccl}(G)$, $m = \degree{\ccl}(\I,F)$, and $n = \degree{\ccl}(\I,G)$ for some $F$, $G$, $\I$, then
	\begin{align*}
	\opt{\ccl}(F \cclConn G) & = k + \ell, \text{ and} \\
	\degree{\ccl}(\I,F \cclConn G) & =
	\begin{cases}
	n & \text{if } m = 1, n < \infty; \\ 
	m + \ell & \text{if } m < \infty \text{ and} \\ & (m > 1 \text{ or } n = \infty); \\
	\infty & \text{otherwise}.
	\end{cases}
	\end{align*}
\end{definition}

Intuitively, given $F \cclConn G$, it is best to satisfy both $F$ and $G$, while satisfying only $F$ is less preferable, but still acceptable. As intended, $\cclConn$ is associative under this new semantics for $\ccl$. This will be shown in Section~\ref{sec:strongEquivalence} (see Lemma~\ref{lemma:orderedDisjunctionIsAssociative}). For a series of distinct propositional variables, we have that $\{a_1, \ldots, a_{(n-i+1)}\} \sat{\ccl}{i} (a_1 \qclConn a_2 \ldots \qclConn a_n)$.

While $\qcl$ and $\ccl$ feature only a single choice connective, our framework now easily allows to define a choice logic 
with the choice connectives of both $\qcl$ and $\ccl$, i.e.\ $\connectives{\qccl} = \{\qclConn, \cclConn\}$, $\optConn{\qccl}{\qclConn} = \optConn{\qcl}{\qclConn}$, $\degreeConn{\qccl}{\qclConn} = \degreeConn{\qcl}{\qclConn}$, and likewise for $\ccl$ and $\cclConn$. 
In this way, a choice logic with several choice connectives can always be seen as the combination of other choice logics.

\begin{table}
	\centering
	\begin{tabular}{ccccccc}
		\toprule
		$\I$ & $a \land b$ & $a \lor b$ & $a \qclConn b$	& $a \cclConn b$	& $a \lclConn b$ & $a \scclConn b$ \\ 
		\midrule
		$\emptyset$ & $\infty$ & $\infty$ & $\infty$ & $\infty$ & $\infty$ & $\infty$ \\ 
		$\{b\}$ & $\infty$ & $1$ & $2$ & $\infty$ & $3$ & $\infty$ \\ 
		$\{a\}$ & $\infty$ & $1$ & $1$ & $2$ & $2$ & $2$ \\ 
		$\{a,b\}$ & $1$ & $1$ & $1$ & $1$ & $1$ & $1$ \\ 
		\bottomrule
	\end{tabular}
	\caption{The classical connectives $\land$, $\lor$ and the choice connectives $\protect\qclConn$ ($\qcl$), $\protect\cclConn$ ($\ccl$), $\protect\lclConn$ ($\lcl$), and $\protect\scclConn$ ($\sccl$) applied to atoms.}
	\label{table:comparisonOfChoiceConnectives}
\end{table}

We now introduce a new logic, called Lexicographic Choice Logic ($\lcl$), whose choice connective deals with satisfaction degrees in a more fine-grained manner.

\begin{definition} \label{def:lcl}
	$\lcl$ is the choice logic such that $\connectives{\lcl} = \{\lclConn\}$, and, if $k = \opt{\lcl}(F)$, $\ell = \opt{\lcl}(G)$, $m = \degree{\lcl}(\I,F)$, and $n = \degree{\lcl}(\I,G)$ for some $F$, $G$, $\I$, then
	\begin{align*}
	\opt{\lcl}(F \lclConn G) & = (k+1) \cdot (\ell+1) - 1, \text{ and} \\
	\degree{\lcl}(\I,F \lclConn G) & = 
	\begin{cases}
	(m-1) \cdot \ell + n & \text{if } m < \infty, n < \infty; \\
	k \cdot \ell + m & \text{if } m < \infty, n = \infty; \\
	k \cdot \ell + k + n & \text{if } m = \infty, n < \infty; \\
	\infty & \text{otherwise}. 
	\end{cases} 
	\end{align*}
\end{definition}

Given $F \lclConn G$ it is best to satisfy both $F$ and $G$, second best to satisfy only $F$, and third best to satisfy only $G$. Satisfying neither $F$ nor $G$ is not acceptable. Let $F = (a \lclConn (b \lclConn c))$. The only interpretation that ascribes a degree of $\infty$ to $F$ is $\emptyset$. The remaining $7$ interpretations applicable to $F$ each result in a different degree, ranging from $1$ to $7$. For example, $\{a,b,c\} \sat{\lcl}{1} F$, $\{a,b\} \sat{\lcl}{2} F$, $\{a,c\} \sat{\lcl}{3} F$, and so on until $\{c\} \sat{\lcl}{7} F$. In this way, $\lclConn$ enables us to succinctly encode lexicographic orderings over variables, which will be useful when analyzing the computational complexity of $\lcl$ in Section~\ref{sec:complexity}. 

It is also possible to define a choice logic that does not use optionality.
We call the following new logic Simple Conjunctive Choice Logic ($\sccl$).

\begin{definition} \label{def:sccl}
	$\sccl$ is the choice logic such that $\connectives{\sccl} = \{\scclConn\}$, and, if $k = \opt{\sccl}(F)$, $m = \degree{\sccl}(\I,F)$, and $n = \degree{\sccl}(\I,G)$ for some $F$, $G$, $\I$, then
	\begin{align*}
	\opt{\sccl}(F \scclConn G) & = k + 1, \text{ and} \\
	\degree{\sccl}(\I,F \scclConn G) & =
	\begin{cases}
	m & \text{if } m < \infty, n < \infty; \\
	m + 1 & \text{if } m < \infty, n = \infty; \\
	\infty & \text{otherwise}.
	\end{cases}
	\end{align*}
\end{definition}

$\sccl$ is similar to $\ccl$ in that the intended meaning of $F\scclConn G$ is the same as that of $F\cclConn G$, i.e., satisfying $F$ and $G$ is most preferable,
but satisfying only $F$ is also acceptable. However, $\sccl$ does not use optionality 
to penalize less preferable interpretations. Instead, the degree of such interpretations is
simply incremented by $1$. 

Consider $F = (a \scclConn (b \scclConn c))$. Then both $\{a,b\}$ and $\{a,b,c\}$ ascribe
a degree of~$1$ to $F$. The fact that $(b \scclConn c)$ is not optimally satisfied by
$\{a,b\}$ is irrelevant in this case. This is in contrast to $\ccl$,
where $(a \cclConn (b \cclConn c))$ would be satisfied to a degree of~$2$ by $\{a,b\}$.
Note that $\scclConn$ is not associative, since $\{a,b\}$ ascribes a degree of~$1$ to $F$ but a degree of~$2$ to $((a \scclConn b) \scclConn c)$.

Lastly, we want to discuss two variants of $\qcl$ introduced by \citet{benferhat2008two} called $\pqcl$ and $\qclPlus$. Both of these logics define $\qclConn$ in the same way as standard $\qcl$, but differ in how classical connectives deal with satisfaction degrees. The alternative conjunctions and disjunctions featured in $\pqcl$ and $\qclPlus$ can be implemented in our framework as choice connectives, if desired. In fact, some behave more like choice connectives than classical connectives, as, for example, the conjunction of $\pqcl$ favors interpretations that ascribe a lower satisfaction degree to the first operand rather than the second. The semantics of $\lclConn$ in $\lcl$ is built on this principle, and is actually an extension of $\land$ in $\pqcl$. The alternative definition for negation in $\pqcl$ and $\qclPlus$ is more involved, but can be simulated in our framework by restricting ourselves to formulas where negations appear only in front of atoms. This means that $\pqcl$ and $\qclPlus$ can be captured by our framework as fragments. 

\subsection{Expressiveness} \label{sec:expressiveness}

It can be shown that any 
logic defined in our framework is powerful
enough to express
arbitrary assignments of satisfaction degrees to interpretations 
by a forumula as long as the degrees are obtainable in the following sense:

\begin{definition}
	A degree $m\in \setOfSatisfactionDegrees$ is called 
	obtainable in a choice logic $\cl$ if there exists 
	an interpretation $\I$ and an $\cl$-formula $F$ such that $\degree{\cl}(\I, F) = m$.
	By $\obtainableDegrees{\cl}$ we denote the set of all degrees obtainable in a choice logic $\cl$. 
\end{definition}

For example, $\obtainableDegrees{\pl} = \{1,\infty\}$ and $\obtainableDegrees{\qcl} = \setOfSatisfactionDegrees$. As soon as a degree $m$ is obtainable, any interpretation can be assigned this degree via a suitable formula. 

\begin{lemma} \label{lemma:synthesisLight}
	Let $m \in \obtainableDegrees{\cl}$ for some choice logic $\cl$. 
	Then there is an $\cl$-formula $F$ such that $\degree{\cl}(\I, F) = m$ for every interpretation $\I$.
\end{lemma}
\begin{proof}
	Let $\cl$ be a choice logic, and let $\I$ be an interpretation. Let $m \in \obtainableDegrees{\cl}$. Since $m$ is obtainable in $\cl$, there is a formula $G$ such that $\degree{\cl}(\J,G) = m$ for some interpretation $\J$. We obtain $F = T(G)$ by transforming $G$ as follows: 
	\begin{enumerate}
		\item $T(a) = 
		\begin{cases}
		\top & \text{if } a \in \J \\
		\bot & \text{otherwise}
		\end{cases}$
		\item $T(\neg A) = \neg T(A)$
		\item $T(A \generalConn B) = T(A) \generalConn T(B)$, where $\generalConn \in \{\land, \lor\} \cup \connectives{\cl}$.
	\end{enumerate}
	The above transformation takes $G$ and replaces every variable $a$ that is contained in $\J$ by $\top$. If $a$ is not contained in $\J$, then it will be replaced by $\bot$. It is easy to see that $\degree{\cl}(\I,F) = \degree{\cl}(\J,G) = m$ since
	\begin{align*}
	\degree{\cl}(\I,T(a)) 
	& = \begin{cases}
	\degree{\cl}(\I, \top) & \text{if } a \in \J \\
	\degree{\cl}(\I, \bot) & \text{otherwise}
	\end{cases}
	\\ & = \begin{cases}
	1 & \text{if } a \in \J \\
	\infty & \text{otherwise}
	\end{cases}
	\\ & = \degree{\cl}(\J,a). \qedhere
	\end{align*}
\end{proof}

With the above lemma we can tackle the issue of formula synthesis:

\begin{proposition} \label{prop:synthesis}
	Let $\cl$ be a choice logic. Let $V$ be a finite set of propositional variables, and let $s$ be a function $s \colon \powerset{V} \to \obtainableDegrees{\cl}$. Then there is an $\cl$-formula $F$ such that for every $\I \subseteq V$, $\degree{\cl}(\I, F) = s(\I)$.
\end{proposition}
\begin{proof}
	Because of Lemma~\ref{lemma:nonOccuringAtomsHaveNoInfluence}, we can assume that every interpretation we are dealing with is a subset of $V$. Let $G_\J$ be the classical formula that characterizes the interpretation $\J$, i.e. 
	\begin{align*}
	G_\J =  \big(
	\bigwedge_{a \in \J} a \big) \land \big(\bigwedge_{a \in (V \setminus \J)} \neg a 						
	\big).
	\end{align*}	 
	Observe that $\J \models G_\J$, but $\J' \not\models G_\J$ for all $\J' \neq \J$. From Lemma \ref{lemma:synthesisLight}, we know that for every $\J$, there is an $\cl$-formula $S_\J$ such that $\degree{\cl}(\J,S_\J) = s(\J)$. Furthermore, let 
	\begin{align*}
	F = \bigvee_{\J \subseteq V} (G_\J \land S_\J).
	\end{align*}
	Let $\I$ be an arbitrary interpretation, and let $C$ be an arbitrary clause in $F$, i.e.\ $C = (G_\J~\land~S_\J)$ for some $\J$. We distinguish two cases:
	\begin{enumerate}
		\item $\I = \J$. Then $\degree{\cl}(\I,G_\J ) = 1$ and $\degree{\cl}(\I,S_\J) = s(\I)$, which implies that $\degree{\cl}(\I,C) = s(\I)$.
		\item $\I \neq \J$. Then $\degree{\cl}(\I,G_\J ) = \infty$ and therefore $\degree{\cl}(\I,C) = \infty$.
	\end{enumerate}
	By construction, there is exactly one clause in $F$ such that $\I = \J$. Since this clause is satisfied with degree $s(\I)$ by $\I$, and all other clauses are ascribed a degree of $\infty$ by $\I$, we have that $\degree{\cl}(\I,F) = s(\I)$. 
\end{proof}

Observe that Proposition~\ref{prop:synthesis} also speaks about interpretations $\I$ that are no subsets of $V$: since $\variablesOf{F} \subseteq V$, $\degree{\cl}(\I,F) = \degree{\cl}(\I \cap V, F) = s(\I \cap V)$.

\section{Strong Equivalence} \label{sec:strongEquivalence}

We now investigate 
strong equivalence in the sense of replaceability \citep{faber2013abstract} of preferred models. $F\replace{A}{B}$ denotes the substitution of an occurrence of $A$ in $F$ by $B$. 

\begin{definition} \label{def:strongEquivalence}
	Two formulas $A$ and $B$ of a choice logic $\cl$ are strongly equivalent, written as $A \strongEquiv{\cl} B$, if $\prefModels{\cl}{F} = \prefModels{\cl}{F\replace{A}{B}}$ for all $\cl$-formulas $F$.
\end{definition}

Strong equivalence is a crucial notion towards simplification, 
thus a characterization that avoids
going through infinitely many formulas is needed.
We show that strong equivalence often can be decided via simpler equivalence notions.

\begin{definition} \label{def:degreeEquivalence}
	Two formulas $A$ and $B$ of a choice logic $\cl$ are degree-equivalent, written as $A \degreeEquiv{\cl} B$, if $\degree{\cl}(\I,A) = \degree{\cl}(\I,B)$ for all interpretations $\I$.
\end{definition}

\begin{lemma} \label{lemma:degreeEquivImpliesSamePrefModels}
	Let $\cl$ be a choice logic. If $A \degreeEquiv{\cl} B$, then $\prefModels{\cl}{A} = \prefModels{\cl}{B}$.
\end{lemma}
\begin{proof}
	Assume $A \degreeEquiv{\cl} B$. Then
	\begin{align*}
	\I \in \prefModels{\cl}{A} 
	\iff & \degree{\cl}(\I,A) \neq \infty \text{ and } \degree{\cl}(\I,A) \leq \degree{\cl}(\J,A) \text{ for all } \J \\
	\iff & \degree{\cl}(\I,B) \neq \infty \text{ and } \degree{\cl}(\I,B) \leq \degree{\cl}(\J,B) \text{ for all } \J \\
	\iff & \I \in \prefModels{\cl}{B}. \qedhere
	\end{align*}
\end{proof}

The converse of Lemma~\ref{lemma:degreeEquivImpliesSamePrefModels} is not true: $a$ and $(a \qclConn b)$ have the same preferred models in $\qcl$, but $\{b\} \sat{\qcl}{\infty} a$, while $\{b\} \sat{\qcl}{2} (a \qclConn b)$. Since the satisfaction degree of a formula might also depend on optionality, the notion of degree-equivalence is not strong enough in many cases; we thus also consider the following notion, which actually appears under the name of strong equivalence by \citet{brewka2004qualitative}.

\begin{definition} \label{def:fullEquivalence}
	Two $\cl$-formulas $A$ and $B$ of a choice logic $\cl$ are fully equivalent, written as $A \fullEquiv{\cl} B$, if  $A \degreeEquiv{\cl} B$ and $\opt{\cl}(A) = \opt{\cl}(B)$.
\end{definition}

\begin{lemma} \label{lemma:fullEquivSubstitution}
	Let $\cl$ be a choice logic. Then $A \fullEquiv{\cl} B$ if and only if $F \fullEquiv{\cl} F\replace{A}{B}$ for all $\cl$-formulas $F$.
\end{lemma}
\begin{proof}
	Assume that $A \fullEquiv{\cl} B$. Then, since the optionality- and degree-functions of any connective in a choice logic are functions over the optionalities and satisfaction degrees of the immediate subformulas, the fact that $F \fullEquiv{\cl} F\replace{A}{B}$ holds for any $\cl$-formula $F$ can be shown easily by structural induction. 
	
	For the converse direction, assume that $A \not\fullEquiv{\cl} B$. Choose $F = A$. Then $F\replace{A}{B} = B$, and therefore $F \not\fullEquiv{\cl} F\replace{A}{B}$.  
\end{proof}

The above lemma states that two formulas are fully equivalent if and only if they 
can be substituted for each other without affecting satisfaction degree or optionality.
This is  even stronger than strong equivalence which only demands that substitution
does not affect preferred models.

\begin{proposition} \label{prop:equivalences}
	${A \fullEquiv{\cl} B} \implies {A \strongEquiv{\cl} B} \implies {A \degreeEquiv{\cl} B}$ for any choice logic $\cl$ and all $\cl$-formulas $A,B$.
\end{proposition}
\begin{proof}
	${A \fullEquiv{\cl} B} \implies {A \strongEquiv{\cl} B}$ follows from Lemma~\ref{lemma:fullEquivSubstitution}. It remains to show that $A \strongEquiv{\cl} B \implies A \degreeEquiv{\cl} B$: assume $A \not\degreeEquiv{\cl} B$. We want to show that there is a formula $F$ such that $\prefModels{\cl}{F} \neq \prefModels{\cl}{F\replace{A}{B}}$. 
	
	Since $A \not\degreeEquiv{\cl} B$, there exists an interpretation $\I$ such that $\I \sat{\cl}{m} A$ and $\I \sat{\cl}{n} B$ with $m \neq n$. Let $k = min(m,n)$. Due to Proposition~\ref{prop:synthesis}, we know that there are formulas $G$ and $H$ such that the minimum degree that satisfies $G$ or $H$ is $k$. Thus, there are interpretations $\I_G$, $\I_H$ such that $\I_G \sat{\cl}{k} G$ and $\I_H \sat{\cl}{k} H$. We can assume that $G$ and $H$ are variable disjoint from each other as well as from $A$ and $B$, since we can always rename variables if necessary. By Lemma~\ref{lemma:nonOccuringAtomsHaveNoInfluence}, we can assume $\I \cap \I_G = \emptyset$, $\I \cap \I_H = \emptyset$, and $\I_G \cap \I_H = \emptyset$. We now construct 
	\begin{align*}
	F = (A \land G) \lor (a \land H),
	\end{align*}
	where $a$ is a fresh variable that does not occur in $A$, $B$, $G$ or $H$. We can therefore also assume that $a$ is not contained in $\I$, $\I_G$, or $\I_H$. 
	
	Observe that the minimal degree with which $F$ (or $F\replace{A}{B}$) can possibly be satisfied is $k$, as either $G$ or $H$ need to be satisfied. Furthermore, $\I_H \cup \{a\} \sat{\cl}{k} F$ and $\I_H \cup \{a\} \sat{\cl}{k} F\replace{A}{B}$. This means that any preferred model of $F$ must satisfy $F$ with a degree of $k$. The same is true for preferred models of $F\replace{A}{B}$. Also observe that since $a$ is not contained in $\I$ or $\I_G$, $\I \cup \I_G \sat{\cl}{\infty} (a \land H)$. We distinguish two cases:	
	\begin{enumerate}
		\item $k = m$. Then $\I \sat{\cl}{k} A$, and therefore $\I \cup \I_G \sat{\cl}{k} (A \land G)$. This implies $\I \cup \I_G \sat{\cl}{k} F$, i.e.\ $\I \cup \I_G \in \prefModels{\cl}{F}$. Analogously, since $\I \sat{\cl}{n} B$, we have that $\I \cup \I_G \sat{\cl}{n} (B \land G)$. Therefore, $\I \cup \I_G \sat{\cl}{n} F\replace{A}{B}$. Since $n > k$, we have $\I \cup \I_G \not\in \prefModels{\cl}{F\replace{A}{B}}$.
		\item $k = n$. Analogous to (1), but with $\I \cup \I_G \not\in \prefModels{\cl}{F}$ and $\I \cup \I_G \in \prefModels{\cl}{F\replace{A}{B}}$.
	\end{enumerate}
	It can be concluded that $\prefModels{\cl}{F} \neq \prefModels{\cl}{F\replace{A}{B}}$, i.e.\ $A \not\strongEquiv{\cl} B$.
\end{proof}

In general, all three equivalence notions are different. For example, in $\sccl$, $a$ and $(a \scclConn a)$ are not fully equivalent, since they differ in optionality, but they are strongly equivalent, since optionality does not impact satisfaction degrees in $\sccl$. On the other hand, in $\qcl$, $a$ and $(a \qclConn a)$ are degree-equivalent and $F = (((a \qclConn b) \lor (c \qclConn d)) \land \neg a \land \neg c)$ has $\{b\}$ as a preferred model, but
replacing the first occurrence of $a$ in $F$ by $a \qclConn a$ means $\{b\}$ is no longer preferred.
However, for natural classes of choice logics strong equivalence coincides with
either degree-equivalence or full equivalence.

\begin{definition} \label{def:optIgnoring}
	A choice logic $\cl$ is called optionality-ignoring if for all $\generalConn \in \connectives{\cl}$,
	$\degree{\cl}(\I,F \generalConn G) = \degree{\cl}(\I,F' \generalConn G')$ whenever
	$\degree{\cl}(\I,F) = \degree{\cl}(\I,F')$ and $\degree{\cl}(\I,G) = \degree{\cl}(\I,G')$.
\end{definition}

It is easy to see that $\pl$ and $\sccl$ are optionality-ignoring, while $\qcl$, $\ccl$, and $\lcl$ are not. 

\begin{proposition} \label{prop:strongEquivForOptIgnoring}
	Let $\cl$ be an optionality-ignoring choice logic. For all $\cl$-formulas $A,B$ we have that $A \strongEquiv{\cl} B$ iff $A \degreeEquiv{\cl} B$.
\end{proposition}
\begin{proof}
	The only-if-direction follows directly from Proposition~\ref{prop:equivalences}. For the if-direction, we can prove that if $A \degreeEquiv{\cl} B$, then $F \degreeEquiv{\cl} F\replace{A}{B}$ for all $\cl$-formulas~$F$, which implies that $\prefModels{\cl}{F} = \prefModels{\cl}{F\replace{A}{B}}$ for all $\cl$-formulas~$F$. This can be done by structural induction, analogous to Lemma~\ref{lemma:fullEquivSubstitution}. The critical difference here is that $\cl$ is optionality-ignoring, and therefore optionality can not influence degrees when substituting $B$ for $A$. 
\end{proof}

Similar to the above result, we can characterize strong equivalence by full equivalence in choice logics where, as soon as two formulas differ in optionality, they can no longer be safely substituted for each other without affecting satisfaction degrees.

\begin{definition} \label{def:optDifferentiating}
	A choice logic $\cl$ is called optionality-differentiating if for all $\cl$-formulas $A$ and $B$ with $\opt{\cl}(A) \neq \opt{\cl}(B)$, there is an $\cl$-formula $F$ such that $F \not\degreeEquiv{\cl} F\replace{A}{B}$.
\end{definition}

$\pl$, $\qcl$, $\ccl$, and $\lcl$ are optionality-differentiating (while $\sccl$ is not). We show this for $\qcl$: Consider any $A,B$ such that $\opt{\qcl}(A) \neq \opt{\qcl}(B)$. Then $F = ((A \land \bot) \qclConn a)$ has the desired property: $(A \land \bot)$ can never be satisfied, and since ${\opt{\qcl}(A \land \bot) = \opt{\qcl}(A)}$, $F$ is satisfied by $\{a\}$ with a degree of $\opt{\qcl}(A) + 1$. Likewise, $F\replace{A}{B}$ is satisfied by $\{a\}$ with a degree of $\opt{\qcl}(B) + 1$. %

\begin{proposition} \label{prop:strongEquivForOptDifferentiating}
	Let $\cl$ be an optionality-differentiating choice logic. For all $\cl$-formulas $A,B$ we have that $A \strongEquiv{\cl} B$ iff $A \fullEquiv{\cl} B$.
\end{proposition}
\begin{proof} %
	The if-direction follows directly from Proposition~\ref{prop:equivalences}. For the only-if-direction, assume that  $A \strongEquiv{\cl} B$. Then, from Proposition~\ref{prop:equivalences}, we know that $A \degreeEquiv{\cl} B$. It remains to show that $\opt{\cl}(A) = \opt{\cl}(B)$. 
	
	Towards a contradiction, assume that $\opt{\cl}(A) \neq \opt{\cl}(B)$. This means that, since $\cl$ is optionality-differentiating, there exists a formula $A' \in \formulas{\cl}$ such that $A' \not\degreeEquiv{\cl} A'\replace{A}{B}$. Observe that therefore $A$ must occur in $A'$. Let $B' = A'\replace{A}{B}$. Then $A' \not\degreeEquiv{\cl} B'$. By the contrapositive of Proposition~\ref{prop:equivalences}, there exists a formula~$F$ such that $\prefModels{\cl}{F} \neq \prefModels{\cl}{F\replace{A'}{B'}}$. By the construction of $F$ in the proof for Proposition~\ref{prop:equivalences}, we can assume that $A'$ occurs only once in $F$, and that $A$ only occurs in $A'$. Therefore, replacing $A'$ by $A'\replace{A}{B}$ in $F$ is the same as simply replacing $A$ by $B$ in $F$, i.e.\ $F\replace{A'}{B'} = F\replace{A'}{A'\replace{A}{B}} = F \replace{A}{B}$. Thus, $\prefModels{\cl}{F} \neq \prefModels{\cl}{F\replace{A}{B}}$. But then  $A \not\strongEquiv{\cl} B$. Contradiction. 
\end{proof}

In fact, strong equivalence and full equivalence are interchangeable \textit{only} for optionality-differentiating choice logics. Consider a logic $\cl$ which is not optionality-differentiating. Then, by definition, there are $\cl$-formulas $A$~and~$B$ such that $\opt{\cl}(A) \neq \opt{\cl}(B)$ while $F \degreeEquiv{\cl} F\replace{A}{B}$ holds for all $\cl$-formulas $F$. 
This implies $A \strongEquiv{\cl} B$.
But $A \not\fullEquiv{\cl} B$, since $\opt{\cl}(A) \neq \opt{\cl}(B)$. 

In Section~\ref{sec:examples}, we 
discussed associativity of $\qclConn$ in $\qcl$. %
A matching strong equivalence result 
can now be achieved for $\ccl$ as well, and it is easily proven using
Proposition~\ref{prop:strongEquivForOptDifferentiating}.

\begin{lemma} \label{lemma:orderedDisjunctionIsAssociative}
	The choice connective  $\cclConn \in \connectives{\ccl}$ is associative, meaning that $((F \cclConn G) \cclConn H) \strongEquiv{\ccl} (F \cclConn (G \cclConn H))$ 
	for any $F,G,H\in
	\formulas{\ccl}$. 
\end{lemma}
\begin{proof}
	For readability, we write $\opt{}$ instead of $\opt{\ccl}$ and $\degree{}$ instead of $\degree{\ccl}$ in this proof.
	First of all, by Definition~\ref{def:cclAlternative}, $\opt{}((F \cclConn G) \cclConn H) = \opt{}(F \cclConn (G \cclConn H))$. We further show that $((F \cclConn G) \cclConn H) \degreeEquiv{\ccl} (F \cclConn (G \cclConn H))$. Let $\I$ be an arbitrary interpretation. Let $d_A = \degree{}(\I, A)$ for $A \in \{F,G,H\}$. We distinguish the following cases:
	\begin{itemize}
		
		\item $d_F = 1$, $d_G = 1$, and $d_H < \infty$. Then $\degree{}(\I, F \cclConn G) = 1$, $\degree{}(\I, G \cclConn H) = d_H$. This implies $\degree{}(\I, (F \cclConn G) \cclConn H) = d_H = \degree{}(\I, F \cclConn (G \cclConn H))$.
		
		\item $d_F = 1$, $d_G = 1$, and $d_H = \infty$. Then $\degree{}(\I, F \cclConn G) = 1$, $\degree{}(\I, G \cclConn H) = 1 + \opt{}(H)$. We can conclude that $\degree{}(\I, (F \cclConn G) \cclConn H) = 1 + \opt{}(H) = \degree{}(\I, F \cclConn (G \cclConn H))$.
		
		\item $d_F = 1$ and $1 < d_G < \infty$. Then $\degree{}(\I, F \cclConn G) = d_G$, $\degree{}(\I, G \cclConn H) = d_G + \opt{}(H)$. This entails that $\degree{}(\I, (F \cclConn G) \cclConn H) = d_G + \opt{}(H) = \degree{}(\I, F \cclConn (G \cclConn H))$.
		
		\item $d_F = 1$ and $d_G = \infty$. Then $\degree{}(\I, F \cclConn G) = 1 + \opt{}(G)$, $\degree{}(\I, G \cclConn H) = \infty$. This further implies $\degree{}(\I, (F \cclConn G) \cclConn H) = 1 + \opt{}(G) + \opt{}(H) = \degree{}(\I, F \cclConn (G \cclConn H))$.
		
		\item $1 < d_F < \infty$. Then $\degree{}(\I, F \cclConn G) = d_F + \opt{}(G)$, which entails that $\degree{}(\I, (F \cclConn G) \cclConn H) = d_F + \opt{}(G) + \opt{}(H) = \degree{}(\I, F \cclConn (G \cclConn H))$.
		
		\item $d_F = \infty$. Then $\degree{}(\I, F \cclConn G) = \infty$, which implies
		$\degree{}(\I, (F \cclConn G) \cclConn H) = \infty = \degree{}(\I, F \cclConn (G \cclConn H))$. \qedhere
	\end{itemize}
\end{proof}

To conclude, knowing whether a choice logic is optionality-ignoring (e.g.\ $\sccl$%
) or optionality-differentiating (e.g.\ $\qcl$, $\ccl$, $\lcl$%
) is useful, since it allows to decide strong equivalence via degree- or full equivalence.
However, note that choice logics which are neither optionality-ignoring nor -differentiating do exist.

\section{Computational Complexity} \label{sec:complexity}

Next, we examine the computational complexity of choice logics.
So far we have imposed only few restrictions on optionality- and degree functions, which means that there are choice logics whose semantics are given by computationally expensive or even undecidable functions.
For reasons of practicality, we focus on so-called tractable choice logics.

\begin{definition}
	A choice logic $\cl$ is called tractable if the optionality- and degree functions of every choice connective in $\cl$ are polynomial-time computable. 
\end{definition}

All logics presented in this paper are tractable in this sense. The first problems that we look at are concerned with satisfaction degrees: 

\begin{definition}
	We define decision problems 
	\begin{itemize}
		\item
		\clModelChecking{$\cl$}: given 
		$F\in\formulas{\cl}$, an interpretation $\I$, and a satisfaction degree $k \in \setOfSatisfactionDegrees$, does $\degree{\cl}(\I,F) \leq k$ hold; %
		\item
		\clSAT{$\cl$}: given 
		$F\in\formulas{\cl}$ and a satisfaction degree $k$, is there an interpretation $\I$ such that $\degree{\cl}(\I,F) \leq k$.
	\end{itemize}
\end{definition}

These two problems are straightforward generalizations of model checking and satisfiability for classical propositional logic, and they do not differ in complexity from their classical counterparts.

\begin{proposition} \label{prop:clModelCheckingMembership}
	\clModelChecking{$\cl$} is in~$\PolyTime$ for any choice logic~$\cl$ that is tractable. 
\end{proposition}
\begin{proof}
	We can compute the degree of $F$ under $\I$ by applying $\degree{\cl}$ to $\I$ and $F$ (see Definition~\ref{def:clSatisfactionDegree}). Since $\cl$ is tractable, we know that every step in this recursion runs in polynomial time. The depth of the recursion can not exceed the length of $F$. Neither can the width of the recursion exceed the length of $F$, since every atom in $F$ will be reached exactly once in the recursion. After we computed $\degree{\cl}(\I,F)$, we can simply compare the result with $k$, and return either "yes" or "no" accordingly. 
\end{proof}

\begin{proposition} \label{prop:clSatMembership}
	\clSAT{$\cl$} is $\NP$-complete for any choice logic~$\cl$ that is tractable. 
\end{proposition}
\begin{proof}
	We prove $\NP$-membership 
	by providing a polynomially balanced and polynomially decidable certificate relation for \clSAT{$\cl$}. Let
	\begin{align*}
	R = \{((F,k), \I) \mid \degree{\cl}(\I,F) \leq k \}.
	\end{align*}
	Clearly, $(F,k)$ is a yes-instance of \clSAT{$\cl$} iff there is an interpretation $\I$ such that $((F,k),\I) \in R$. $R$ is polynomially balanced, since we can assume that $\I \subseteq \variablesOf{F}$. Furthermore, $R$ is polynomially decidable, since \clModelChecking{$\cl$} is in~$\PolyTime$. 
	
	For $\NP$-hardness, we provide a reduction from \SAT. Let $F$ be an arbitrary classical formula. Then, for any interpretation $\I$ it holds that $\I \models F \iff \I \sat{\cl}{1} F \iff \degree{\cl}(\I,F) \leq 1$. Thus, 
	\begin{align*}
	& F \text{ is a yes-instance of \SAT} \\
	& \iff \text{there is an interpretation $\I$ such that } \I \models F \\
	& \iff \text{there is an interpretation $\I$ such that } \degree{\cl}(\I,F) \leq 1 \\
	& \iff (F,1) \text{ is a yes-instance of \clSAT{$\cl$}}. \qedhere
	\end{align*}
\end{proof}

The picture changes when we reformulate the problems in terms of 
preferred models.

\begin{definition}
	We define decision problems 
	\begin{itemize}
		\item
		\clPrefModelChecking{$\cl$}: given 
		$F\in\formulas{\cl}$
		and an interpretation $\I$, is $\I \in \prefModels{\cl}{F}$; %
		\item
		\clPrefModelSAT{$\cl$}: given 
		$F\in\formulas{\cl}$
		and variable $a$, is there an interpretation $\I \in \prefModels{\cl}{F}$ such that $a\in \I$.
	\end{itemize}
\end{definition}

\begin{proposition} \label{prop:clPrefModelCheckingMembership}
	\clPrefModelChecking{$\cl$} is in $\coNP$ for any tractable choice logic $\cl$.
\end{proposition}
\begin{proof}
	We will show that the complementary problem is in $\NP$ by providing a polynomially balanced, polynomially decidable certificate relation. First of all, $(F,\I)$ is a yes-instance of co-\clPrefModelChecking{$\cl$} iff $\I \not\in \prefModels{\cl}{F}$. Let 
	\begin{align*}
		R = \{((F,\I), \J) \mid  \degree{\cl}(\J,F) < \degree{\cl}(\I,F) \text{ or } \degree{\cl}(\I,F) = \infty \}.
	\end{align*}
	Then $((F,\I), \J) \in R$ iff  there is an interpretation $\J$ such that $\degree{\cl}(\J,F) < \degree{\cl}(\I,F)$ or $\degree{\cl}(\I,F) = \infty$, which is the case exactly when $\I \not\in \prefModels{\cl}{F}$. $R$ is polynomially balanced since we can assume that $\J \subseteq \variablesOf{F}$. Furthermore, $R$ is polynomially decidable, since we can compute $\degree{\cl}(\I,F)$ and $\degree{\cl}(\J,F)$ in polynomial time for tractable choice logics. 
\end{proof}

We observe that \clPrefModelChecking{$\pl$} and \clModelChecking{$\pl$} are identical, since ${\I \in \prefModels{\pl}{F} \iff \I \sat{\pl}{1} F}$. 
Hence, we cannot expect $\coNP$-hardness in general, 
but we 
show the result for all logics where degrees other than $1$ and $\infty$ are obtainable, and thus for
$\qcl, \ccl, \sccl,$ and $\lcl$. %

\begin{proposition} \label{prop:clPrefModelCheckingHardness}
	\clPrefModelChecking{$\cl$} is 
	$\coNP$-complete
	for any tractable choice logic $\cl$ where $\obtainableDegrees{\cl} \neq \{1,\infty\}$.
\end{proposition}
\begin{proof}
	$\coNP$-membership can be inferred from Proposition~\ref{prop:clPrefModelCheckingMembership}.
	We show hardness by providing a reduction from an arbitrary instance $F$ of \Unsat. Let $m \in \obtainableDegrees{\cl} \setminus \{1,\infty\}$, and let $a$ be a variable that does not occur in $F$. By Proposition~\ref{prop:synthesis}, there exists an $\cl$-formula $G$ such that $\I \sat{\cl}{m} G$ if $a \in \I$ and $\I \sat{\cl}{\infty} G$ if $a \not\in \I$ for all interpretations $\I$. Note that the size of $G$ is constant with respect to the size of $F$. We now construct an instance $(F',\{a\})$ of \clPrefModelChecking{$\cl$}, where
	\begin{align*}
	F' = (F \lor G) \land \neg (F \land G).
	\end{align*}
	It remains to prove that $F$ is a yes-instance of \Unsat\ if and only if $(F',\{a\})$ is a yes-instance of \clPrefModelChecking{$\cl$}:
	
	"$\implies$": Assume $F$ is a yes-instance of \Unsat. Then there is no interpretation $\J$ such that $\J \models F$, i.e.\ $\degree{\cl}(\J,F) = \infty$ for all $\J$. Since $\degree{\cl}(\{a\},G) = m$, we have that $\degree{\cl}(\{a\},F') = m$. Indeed, $F'$ can not be satisfied to a degree lower than $m$ since $F$ is unsatisfiable and since $m$ is the lowest degree with which $G$ can be satisfied. Thus, $\{a\} \in \prefModels{\cl}{F'}$.
	
	"$\impliedby$": We proceed by contrapositive. Assume $F$ is a no-instance of \Unsat. Then there is an interpretation $\J$ such that $\J \models F$. Because $a$ does not occur in $F$, we can assume that $a \not\in \J$, and therefore $\J \sat{\cl}{\infty} G$. Thus, $\J \sat{\cl}{1} (F \lor G)$ and $\J \sat{\cl}{1} \neg (F \land G)$, which implies that  $\J \sat{\cl}{1} F'$. Recall that $\{a\} \sat{\cl}{m} G$. We distinguish two cases:
	\begin{enumerate}
		\item $\{a\} \models F$. Then $\{a\} \sat{\cl}{1} (F \lor G)$, $\{a\} \sat{\cl}{\infty} \neg (F \land G)$, and therefore $\{a\} \sat{\cl}{\infty} F'$.
		\item $\{a\} \not\models F$.  Then $\{a\} \sat{\cl}{m} (F \lor H)$, $\{a\} \sat{\cl}{1} \neg (F \land H)$, and therefore $\{a\} \sat{\cl}{m} F'$.
	\end{enumerate}
	In any case, $\degree{\cl}(\{a\},F') > \degree{\cl}(\J,F')$ which implies that $\{a\} \not\in \prefModels{\cl}{F'}$.
\end{proof}

We turn to \clPrefModelSAT{$\cl$} %
and first
give an upper bound for the optionality of choice logic formulas relative to their size. 
In the following, $\sizeof{F}$ denotes the total number of variables occurrences in $F$, e.g.\ $\sizeof{(x \land x \land y)} = 3$.

\begin{lemma} \label{lemma:optionalitySizeBound}
	Let $\cl$ be a choice logic. Then, for every $\cl$-formula $F$ it holds that $\opt{\cl}(F) < 2^{(\sizeof{F}^2)}$.
\end{lemma}
\begin{proof}
	By structural induction over $\formulas{\cl}$.
	\begin{itemize}
		
		\item Base case. $F = a$, where $a$ is a propositional variable. Then $\sizeof{F} = 1$ and $\opt{\cl}(F) = 1 < 2^{(\sizeof{F}^2)}$.
		
		\item Step case. As our I.H., let $G$ and $H$ be $\cl$-formulas such that $\opt{\cl}(G) < 2^{(\sizeof{G}^2)}$ and $\opt{\cl}(H) < 2^{(\sizeof{H}^2)}$. We distinguish the following cases:
		\begin{enumerate}
			\item $F = (\neg G)$. Then $\sizeof{F} = \sizeof{G} \geq 1$ and $\opt{\cl}(F) = 1 < 2^{(\sizeof{F}^2)}$.
			\item $F = (G \land H)$ or $F = (G \lor H)$. Then $\sizeof{F} = \sizeof{G} + \sizeof{H}$ and 
			\begin{align*}
			\opt{\cl}(F) 
			& = \max(\opt{\cl}(G),\opt{\cl}(H)) \\
			& < \max(2^{(\sizeof{G}^2)},2^{(\sizeof{H}^2)}) 
			< 2^{(\sizeof{F}^2)}.
			\end{align*}
			\item $F = (G \generalConn H)$, where $\generalConn \in \connectives{\cl}$. Then $\sizeof{F} = \sizeof{G} + \sizeof{H}$. Of course, $\opt{\cl}(G) < 2^{(\sizeof{G}^2)}$ is the same as $\opt{\cl}(G) \leq 2^{(\sizeof{G}^2)} - 1$. Likewise for~$H$. Thus, 
			\begin{align*}
			\opt{\cl}(F) 
			& \leq (\opt{\cl}(G) + 1) \cdot (\opt{\cl}(H) + 1) \\
			& \leq ((2^{(\sizeof{G}^2)} - 1) + 1) \cdot ((2^{(\sizeof{H}^2)} - 1) + 1) \\
			& = 2^{(\sizeof{G}^2)} \cdot 2^{(\sizeof{H}^2)} 
			= 2^{(\sizeof{G}^2) + (\sizeof{H}^2)} \\
			& < 2^{((\sizeof{G} + \sizeof{H})^2)} 
			= 2^{(\sizeof{F}^2)}. \qedhere
			\end{align*}
		\end{enumerate}
	\end{itemize}
\end{proof}

This is likely not a tight bound, but it is good enough for our purposes.
Recall that $\phDelta{2}$ is the class of decision problems that can be solved in polynomial time on a deterministic Turing machine with access to an arbitrary number of $\NP$-oracle calls. $\phTheta{2}$ is defined analogously %
but only a logarithmic number of $\NP$-oracle calls is permitted.

\begin{proposition} \label{prop:delta2Membership}
	\clPrefModelSAT{$\cl$} is in $\phDelta{2}$ and $\NP$-hard for any tractable choice logic $\cl$.
\end{proposition}
\begin{proof}
	$\phDelta{2}$-membership: let $(F,x)$ be an instance of \clPrefModelSAT{$\cl$}. We provide a decision procedure which runs in polynomial time with respect to $\sizeof{F}$, makes $\bigO(\sizeof{F}^2)$ calls to an $\NP$-oracle, and determines whether $(F,x)$ is a yes-instance of \clPrefModelSAT{$\cl$}:
	\begin{enumerate}
		
		\item Construct the formula $F'$ by replacing every occurrence of $x$ in $F$ by the tautology $\top$. Observe that $F$ can be satisfied to a degree of $k$ by some interpretation containing $x$ if and only if $F'$ can be satisfied to the degree of $k$ by any interpretation. Also note that $\sizeof{F'} \in \bigO(\sizeof{F})$ and $\opt{\cl}(F') = \opt{\cl}(F)$. 
		
		\item Conduct a binary search over $(1,\ldots,\opt{\cl}(F),\infty)$. In each step of the binary search, we call an $\NP$-oracle that decides \clSAT{$\cl$} to check whether there is an interpretation $\I$ such that $\degree{\cl}(\I,F) \leq k$, where $k$ is the mid-point of the current step in the binary search. In the end, we will find the minimum $k$ such that $\degree{\cl}(\J,F) = k$ for some $\J$. By Lemma~\ref{lemma:optionalitySizeBound}, $\opt{\cl}(F)~<~2^{(\sizeof{F}^2)}$. Since binary search runs in logarithmic time, we require at most $\bigO(\log(\opt{\cl}(F))) = \bigO(\log(2^{(\sizeof{F}^2)})) = \bigO(\sizeof{F}^2)$ oracle calls.
		
		\item Conduct a binary search over $(1,\ldots,\opt{\cl}(F'),\infty)$ to find the minimum $k'$ such that $\degree{\cl}(\J,F') = k'$ for some $\J$. As before, this requires $\bigO(\sizeof{F'}^2) = \bigO(\sizeof{F}^2)$ $\NP$-oracle calls.
		
		\item If $k < k'$ we have a no-instance and if $k = k'$ we have a yes-instance. Note that it can not be that $k > k'$.
	\end{enumerate}
	
	$\NP$-hardness: let $F$ be an arbitrary instance of \SAT. We then construct an instance $(F',a)$ of \clPrefModelSAT{$\cl$}, where $a$ does not occur in $F$, and $F' = F \land a$. Since $F$ and $F'$ are classical formulas,  and since $a$ does not occur in $F$, we have that $F$ is a yes-instance of \SAT\
	\begin{align*}
	& \iff \text{there is an interpretation $\I$ such that } \I \models F \\
	& \iff \text{there is an interpretation $\I$ such that } \I \cup \{a\} \models F' \\
	& \iff \I \cup \{a\}  \in \prefModels{\cl}{F'} \\
	& \iff (F',a) \text{ is a yes-instance of \clPrefModelSAT{$\cl$}}. \qedhere
	\end{align*}
\end{proof}

These are tight bounds in the sense that there are choice logics for which 
\clPrefModelSAT{$\cl$} is $\phDelta{2}$-complete (Proposition~\ref{prop:delta2CompletenessLCL})
and choice logics for which it is $\NP$-complete (just take $\cl = \pl$).
However, there are also choice logics for which the complexity lies between these two classes. The key is to restrict optionality.

\begin{proposition} \label{prop:theta2Membership}
	\clPrefModelSAT{$\cl$} is in $\phTheta{2}$ for all tractable choice logics~$\cl$ in which for some constant $c$ and all $\cl$-formulas $F$ it holds that $\opt{\cl}(F) \in \bigO(\sizeof{F}^c)$.
\end{proposition}
\begin{proof}
	Analogous to the proof of Proposition~\ref{prop:delta2Membership}, except for the crucial difference that we require at most $\bigO(\log(\opt{\cl}(F))) = \bigO(\log(\sizeof{F}^c)) = \bigO(\log(\sizeof{F}))$ oracle calls.
\end{proof}

This
implies that \clPrefModelSAT{$\cl$} is in $\phTheta{2}$ for $\cl \in \{\qcl, \ccl, \sccl\}$. 
In fact, %
these logics are
$\phTheta{2}$-complete. %

\begin{proposition} \label{prop:theta2CompletenessQCLandCCL}
	Let $\cl \in \{\qcl, \ccl, \sccl \}$. Then, \clPrefModelSAT{$\cl$} is $\phTheta{2}$-complete. 
\end{proposition}
\begin{proof} %
	$\phTheta{2}$-membership is established in Proposition~\ref{prop:theta2Membership}. For $\phTheta{2}$-hardness we provide a reduction from \LogLexMaxSat\ \citep{woltran2018hard} which is known to be $\phTheta{2}$-complete.
	Consider an instance of \LogLexMaxSat, where, given a $\pl$-formula $F$ and an ordering ${x_1 > \cdots > x_n}$ over $n$ of the variables in $F$ with ${n \leq \log(\sizeof{F})}$, we ask whether $x_n$ is true in the lexicographically largest interpretation $\J \subseteq \{x_1, \ldots, x_n\}$ that can be extended to a model of $F$. 
	
	We construct an instance $(F',x_n)$ of \clPrefModelSAT{$\qcl$} as follows: Let $\J_i$ be the lexicographically $i$-th largest interpretation over $x_1 > \cdots > x_n$. For example, $\J_1 = \{x_1,\ldots,x_n\}$, $\J_2 = \{x_1,\ldots,x_{n-1}\}$, and $\J_{(2^n)} = \emptyset$. We characterize each of those interpretations by a formula, namely
	\begin{align*}
	A_i = \Big( \bigwedge_{x \in \J_i} x \Big) \land \Big( \bigwedge_{x \in \{x_1,\ldots,x_n\} \setminus \J_i} \neg x \Big).
	\end{align*}
	Then, for any interpretation $\I$, we have that $\I \models A_i \iff \I \cap \{x_1,\ldots,x_n\}= \J_i$. Now let
	\begin{align*}
	F' = F \land (A_1 \qclConn A_2 \qclConn \cdots \qclConn A_{(2^n)}).
	\end{align*}
	Observe that this construction is polynomial in $\sizeof{F}$, as $n \leq \log(\sizeof{F})$, and therefore $2^n \leq \sizeof{F}$. It remains to show that $(F,(x_1,\ldots,x_n))$ is a yes-instance of \LogLexMaxSat\ iff $(F',x_n)$ is a yes-instance of \clPrefModelSAT{$\qcl$}. 
	
	"$\implies$": Let $(F,(x_1,\ldots,x_n))$ be a yes-instance of \LogLexMaxSat. Then there exists an interpretation $\I$ such that $x_n \in \I$, $\I \models F$, and such that $\J_k = \I \cap \{x_1,\ldots,x_n\}$ is the lexicographically largest interpretation over $x_1 > \cdots > x_n$ that can be extended to a model of $F$. Observe that $\I \models A_k$, but $\I \not\models A_r$ for any $r \neq k$. Therefore, by the semantics of ordered disjunction in $\qcl$, $\degree{\qcl}(\I,F') = k$. Let $\I'$ be any interpretation other than $\I$. If $\I' \not\models F$, then $\degree{\qcl}(\I',F') = \infty$. If $\I' \models F$, then it can not be that  $\J_{k'} = \I' \cap \{x_1,\ldots,x_n\}$ is lexicographically larger than $\J_k$ with respect to $x_1 > \cdots > x_n$. Thus, $k \leq k'$. By the same reasoning as above, $\degree{\qcl}(\I',F') = k'$. This means that there is no interpretation that satisfies $F'$ to a smaller degree than $\I$, i.e.\ $\I \in \prefModels{\qcl}{F'}$. Since also $x_n \in \I$, we can conclude that $(F',x_n)$ is a yes-instance of \clPrefModelSAT{$\qcl$}.
	
	"$\impliedby$": Let $(F',x_n)$ be a yes-instance of \clPrefModelSAT{$\qcl$}. Then there is an interpretation $\I$ such that $x_n \in \I$ and $\I \in \prefModels{\qcl}{F'}$. By the construction of $F'$, we have that $\I \models F$. Towards a contradiction, assume there is an interpretation $\I'$ such that $\I' \models F$, and such that $\J_{k'} = \I' \cap \{x_1,\ldots,x_n\}$ is lexicographically larger than $\J_k = \I \cap \{x_1,\ldots,x_n\}$ with respect to $x_1 > \cdots > x_n$. Then $k' < k$. But by the same argument as in the only-if-direction, we can conclude that $\degree{\qcl}(\I,F') = k$ and $\degree{\qcl}(\I',F') = k'$, i.e.\ $\degree{\qcl}(\I,F') < \degree{\qcl}(\I',F')$. But then $\I$ is not a preferred model of $F'$. Contradiction. This means that $\J_k$ is the lexicographically largest interpretation with respect to $(x_1,\ldots,x_n)$ that can be extended to a model of $F$. Therefore, $(F,(x_1,\ldots,x_n))$ is a yes-instance of \LogLexMaxSat.
	
	The proof for \clPrefModelSAT{$\ccl$} is similar. Again, we provide a reduction from \LogLexMaxSat. Let $(F,(x_1,\ldots,x_n))$ be an arbitrary instance of \LogLexMaxSat. We construct an instance $(F',x_n)$ of \clPrefModelSAT{$\ccl$}: as above, $A_i$ is the formula that characterizes the lexicographically $i$-th largest interpretation $\J_i$ over $x_1 > \cdots > x_n$. We further construct 
	\begin{align*}
	C_i = \bigvee_{j=1}^{2^n-(i-1)} A_j.
	\end{align*}
	for every $1 \leq i \leq 2^n$. Then $\J_i \models C_j$ for $1 \leq j \leq  2^n-(i-1)$, and $\J_i \not\models C_j$ for $j > 2^n-(i-1)$. For example, $\J_1$ satisfies $C_1$,\ldots,$C_{(2^n)}$, $\J_2$ satisfies $C_1$,\ldots,$C_{(2^n-1)}$ but not $C_{(2^n)}$, and $\J_{(2^n)}$ satisfies only $C_1$. Let 
	\begin{align*}
	F' = F \land (C_1 \cclConn C_2 \cclConn \cdots \cclConn C_{(2^n)}).
	\end{align*}
	This construction is still polynomial in $\sizeof{F}$: Recall that $n \leq \log(\sizeof{F})$, and therefore $2^n \leq \sizeof{F}$. For every $1 \leq i \leq 2^n$ we have that $\sizeof{A_i} \leq \log(\sizeof{F})$, and thus $\sizeof{C_i} \leq \log(\sizeof{F}) \cdot \sizeof{F}$. This means that $\sizeof{F'} \leq \sizeof{F} + \log(\sizeof{F}) \cdot \sizeof{F}^2$. Also note that, by the semantics of ordered conjunction in $\ccl$, we have that $\degree{\cl}(\J_i, C_1 \cclConn C_2 \cclConn \cdots \cclConn C_{(2^n)}) = i$ for all $1 \leq i \leq 2^n$. Thus, by the same argument as for $\qcl$, we can conclude that $(F,(x_1,\ldots,x_n))$ is a yes-instance of \LogLexMaxSat\ if and only if $(F',x_n)$ is a yes-instance of \clPrefModelSAT{$\ccl$}. 
	
	The proof for $\sccl$ is analogous to $\ccl$, except that we must construct 
	\begin{align*}
	F' = F \land ((((C_1 \scclConn C_2) \scclConn C_3 ) \scclConn \cdots ) \scclConn C_{(2^n)})
	\end{align*}
	since $\scclConn$ is not associative.
\end{proof}

In contrast, $\lcl$ allows to represent the lexicographic order with an exponentially smaller formula
than \qcl, \ccl, and \sccl. 

\begin{lemma} \label{lemma:encodeLexOrderingInLCL}
	Let $x_1 > \cdots > x_n$ be an ordering over $n$ propositional variables. Let $\I_k \subseteq \{x_1, \ldots, x_n\}$ be the lexicographically $k$-th largest interpretation over this ordering, and let $F_n = (x_1 \lclConn (x_2 \lclConn (\cdots (x_{n-1} \lclConn x_n))))$. Then
	\begin{align*}
	\degree{\lcl}(\I_k,F_n) = \begin{cases}
	k & \text{if } k < 2^n \\
	\infty & \text{if } k = 2^n
	\end{cases}
	\end{align*}
\end{lemma}
\begin{proof}
	First, we show $\opt{\lcl}(F_n) = 2^n - 1$ by induction over $n$:
	\begin{itemize}
		\item Base case: $n = 1$. Then $\opt{\lcl}(F_n) = \opt{\lcl}(x_1) = 1 = 2^n - 1$.
		\item Step case: $n > 1$. Then $F_n = (x_1 \lclConn A)$, where $A = (x_2 \lclConn (\cdots (x_{n-1} \lclConn x_n)))$. By our I.H., $\opt{\lcl}(A) = 2^{n-1} - 1$, and thus
		\begin{align*}
		\opt{\lcl}(F_n) 
		& = \opt{\lcl}(x_1 \lclConn A) \\
		& = (\opt{\lcl}(x_1) + 1) \cdot (\opt{\lcl}(A) + 1) - 1 \\
		& = (1 + 1) \cdot ((2^{n-1}-1) + 1) - 1 \\
		& = (2 \cdot 2^{n-1}) - 1 
		= 2^n - 1.
		\end{align*}
	\end{itemize}
	Now we proceed with the main proof, again by induction over $n$:
	\begin{itemize}
		\item Base case: $n = 1$. Then $F_n = x_1$, $\I_1 = \{x_1\}$, and $\I_2 = \emptyset$. Clearly, $\degree{\lcl}(\I_1,F_n) = 1$ and $\degree{\lcl}(\I_2,F_n) = \infty$, as required.
		\item Step case: $n > 1$. Then $F_n = (x_1 \lclConn A)$, where $A = (x_2 \lclConn (\cdots (x_{n-1} \lclConn x_n)))$. Let $\J_\ell$ be the lexicographically $\ell$-th largest interpretation over $x_2 > \cdots > x_n$. By our I.H., $\degree{\lcl}(\J_\ell,A) = \ell$ if $\ell < 2^{n-1}$, and $\degree{\lcl}(\J_\ell,A) = \infty$ if $\ell = 2^{n-1}$. We can obtain the $k$-th largest interpretation $\I_k$ over $x_1 > \cdots > x_n$ as follows:
		\begin{enumerate}
			\item If $k \leq 2^{n-1}$, then $\I_k = \J_k \cup \{x_1\}$. Therefore, $\degree{\lcl}(\I_k,x_1) = 1$. There are two cases:
			\begin{enumerate}
				\item $\degree{\lcl}(\J_k,A) < \infty$. Then $\degree{\cl}(\I_k,F_n) = \degree{\cl}(\J_k,A) = k$.
				\item $\degree{\lcl}(\J_k,A) = \infty$. Then $k = 2^{n-1}$. Thus, \\ $\degree{\cl}(\I_k,F_n) = \opt{\lcl}(A) + 1 = (2^{n-1} - 1) + 1 = 2^{n-1} = k$.
			\end{enumerate}
			\item If $k > 2^{n-1}$, then $\I_{k} = \J_{(k - 2^{n-1})}$. Therefore, $\degree{\lcl}(\I_k,x_1) = \infty$. There are two cases:
			\begin{enumerate}
				\item $\degree{\lcl}(\J_{(k - 2^{n-1})},A) < \infty$. Then $\degree{\cl}(\I_k,F_n) = \opt{\lcl}(A) + 1 + \degree{\cl}(\J_{(k - 2^{n-1})},A) = (2^{n-1} - 1) + 1 + (k - 2^{n-1}) = k$.
				\item $\degree{\lcl}(\J_{(k - 2^{n-1})},A) = \infty$. Then $k - 2^{n-1} = 2^{n-1}$, and therefore $k = 2^n$. Furthermore, $\degree{\cl}(\I_k,F_n) = \infty$, as required. \qedhere
			\end{enumerate}
		\end{enumerate}
	\end{itemize}
\end{proof}

\begin{proposition}\label{prop:delta2CompletenessLCL}
	\clPrefModelSAT{$\lcl$} is $\phDelta{2}$-complete. 
\end{proposition}
\begin{proof}
	$\phDelta{2}$-membership can be inferred from Proposition~\ref{prop:delta2Membership}. We prove $\phDelta{2}$-hardness by providing a reduction from \LexMaxSat\ \citep{woltran2018hard}, where, given a $\pl$-formula $F$ and an ordering ${x_1 > \cdots > x_n}$ over all variables in $F$, we ask whether $x_n$ is true in the lexicographically largest model of $F$. We construct an instance $(F',x_n)$ of \clPrefModelSAT{$\lcl$}, where
	\begin{align*}
	F' = F \land (x_1 \lclConn (x_2 \lclConn (\cdots (x_{n-1} \lclConn x_n)))).
	\end{align*}
	It remains to show that $(F,(x_1,\ldots,x_n))$ is a yes-instance of \LexMaxSat\ if and only if $(F',x_n)$ is a yes-instance of \clPrefModelSAT{$\lcl$}. 
	
	"$\implies$": Let $(F,(x_1,\ldots,x_n))$ be a yes-instance of \LexMaxSat. Then there exists an interpretation $\I$ such that $\I \models F$, $x_n \in \I$, and $\I$ is the lexicographically largest model of $F$ with respect to the ordering $x_1 > \cdots > x_n$. Let $\J$ be any interpretation other than $\I$. If $\J \not\models F$, then $\degree{\lcl}(\J,F') = \infty$, and $\J$ is not a preferred model of $F'$. If $\J \models F$, then $\J$ must be lexicographically smaller than $\I$. By Lemma~\ref{lemma:encodeLexOrderingInLCL}, we can directly infer that $\degree{\lcl}(\I,x_1 \lclConn (\cdots (x_{n-1} \lclConn x_n))) < \degree{\lcl}(\J,x_1 \lclConn (\cdots (x_{n-1} \lclConn x_n)))$, and therefore $\degree{\lcl}(\I,F') < \degree{\lcl}(\J,F')$. This means that $\I \in \prefModels{\lcl}{F'}$. Since also $x_n \in \I$, we have that $(F',x_n)$ is a yes-instance of \clPrefModelSAT{$\lcl$}. 
	
	"$\impliedby$": Let $(F',x_n)$ be a yes-instance of \clPrefModelSAT{$\lcl$}. Then there is an interpretation $\I$ such that $x_n \in \I$ and $\I \in \prefModels{\lcl}{F'}$. Towards a contradiction, assume that there is an interpretation $\J$ such that $\J \models F$, and such that $\J$ is lexicographically larger than $\I$ with respect to $x_1 > \cdots > x_n$. But, by Lemma~\ref{lemma:encodeLexOrderingInLCL}, this means that $\degree{\lcl}(\J,F') < \degree{\lcl}(\I,F')$, which means that $\I \not\in\prefModels{\lcl}{F'}$. Contradiction. Thus, $\I$ is the lexicographically largest model of $F$. Since $x_n \in \I$, we have that $(F,(x_1,\ldots,x_n))$ is a yes-instance of \LexMaxSat.
\end{proof}

Finally, we consider
\clFullEquivChecking{$\cl$}, the problem of deciding whether $A \fullEquiv{\cl} B$ holds for given $\cl$-formulas $A$ and $B$, as well as \clDegreeEquivChecking{$\cl$} and \clStrongEquivChecking{$\cl$} which are defined analogously. In the following result, hardness follows from $\pl$; membership 
of (2) is by our characterizations in 
Section~\ref{sec:strongEquivalence}. %

\begin{proposition}\label{prop:coNPcompletenessOfDegreeEquivalence}
	For any tractable choice logic $\cl$,
	(1)
	\clFullEquivChecking{$\cl$} and \clDegreeEquivChecking{$\cl$} are 
	$\coNP$-complete;
	(2)
	\clStrongEquivChecking{$\cl$} is
	$\coNP$-complete if $\cl$ is
	optionality-ignoring %
	or %
	-differentiating. %
\end{proposition}
\begin{proof}
	We prove $\NP$-membership of co-\clFullEquivChecking{$\cl$}: Let 
	\begin{align*}
		R = \{((A,B), \I) \mid \degree{\cl}(\I,A) \neq \degree{\cl}(\I,B) \text{ or } \opt{\cl}(A) \neq \opt{\cl}(B) \}.
	\end{align*}
	Clearly, $(A,B)$ is a yes-instance of co-\clFullEquivChecking{$\cl$} iff there is an interpretation $\I$ such that $((A,B),\I) \in R$. Furthermore, $R$ is polynomially balanced, since we can assume that $\I \subseteq (\variablesOf{A} \cup \variablesOf{B})$. $R$ is also polynomially decidable, since $\cl$ is a tractable choice logic. 
	
	As for $\coNP$-hardness of \clFullEquivChecking{$\cl$}, let $F$ be an arbitrary instance of \Unsat. We construct an instance $(A,B)$ of \clFullEquivChecking{$\cl$} with $A = F$ and $B = \bot$. Observe that $\opt{\cl}(A) = \opt{\cl}(B) = 1$, since $F$ is a classical formula. Additionally, $A$ and $B$ are degree-equivalent if and only if $\I \not\models A$ for all $\I$, which is the case if and only if $F = A$ is a yes-instance of \Unsat.
\end{proof}

Table~\ref{table:complexityResults1} and Table~\ref{table:complexityResults2}
summarize our
complexity results 
for tractable choice logics, thus including all specific logics studied so far; a full analysis may be conducted in a similar way by using oracles for the optionality- and satisfaction degree functions. %

\begin{table*}
	\centering
	\begin{tabular}{lccc}
		\toprule
		& $\cl = \pl$ & $\cl \in \{\qcl,\ccl, \sccl\}$ & $\cl = \lcl$ \\
		\midrule
		\clModelChecking{$\cl$} & in $\PolyTime$ & in $\PolyTime$ &  in $\PolyTime$ \\
		\clSAT{$\cl$} & $\NP$-c & $\NP$-c & $\NP$-c \\
		\clPrefModelChecking{$\cl$} & in $\PolyTime$ & $\coNP$-c & $\coNP$-c \\
		\clPrefModelSAT{$\cl$} & $\NP$-c & $\phTheta{2}$-c & $\phDelta{2}$-c \\
		\clFullEquivChecking{$\cl$} & $\coNP$-c  & $\coNP$-c & $\coNP$-c \\
		\clDegreeEquivChecking{$\cl$} & $\coNP$-c  & $\coNP$-c & $\coNP$-c \\
		\clStrongEquivChecking{$\cl$} & $\coNP$-c  & $\coNP$-c & $\coNP$-c \\
		\bottomrule
	\end{tabular}
	\caption{Summary of complexity results 1}
	\label{table:complexityResults1}
\end{table*}
\begin{table*}
	\centering
	\begin{tabular}{lccc}
		\toprule
		& Opt-diff. & Opt-ignor. & Tract. \\
		\midrule
		\clModelChecking{$\cl$} &  in $\PolyTime$ &  in $\PolyTime$ &  in $\PolyTime$ \\
		\clSAT{$\cl$} & $\NP$-c & $\NP$-c & $\NP$-c \\
		\clPrefModelChecking{$\cl$} & in $\coNP$ & in $\coNP$ & in $\coNP$\\
		\clPrefModelSAT{$\cl$} & in $\phDelta{2}$/$\NP$-h & in $\phDelta{2}$/$\NP$-h & in $\phDelta{2}$/$\NP$-h \\
		\clFullEquivChecking{$\cl$} & $\coNP$-c & $\coNP$-c & $\coNP$-c \\
		\clDegreeEquivChecking{$\cl$} & $\coNP$-c & $\coNP$-c & $\coNP$-c \\
		\clStrongEquivChecking{$\cl$} & $\coNP$-c  & $\coNP$-c  & ? \\
		\bottomrule
	\end{tabular}
	
	\caption{Summary of complexity results 2.}
	\label{table:complexityResults2}
\end{table*}

\section{Conclusion} \label{sec:conclusion}

We defined and investigated a general framework for choice logics that captures both $\qcl$ and $\ccl$, but also allows to define new logics %
as examplified via
$\sccl$ and $\lcl$. %
We have shown that strong equivalence %
is interchangeable with degree-equivalence for optionality-ignoring choice logics (e.g.\ $\sccl$) and with full equivalence for optionality-differentiating choice logics (e.g.\ $\qcl$, $\ccl$, $\lcl$). Moreover, the computational complexity of tractable choice logics was investigated in detail. An initial definition of our framework and some further results regardings choice logics can be found in the master thesis of the first author \citep{bernreiter2020thesis}. Moreover, ASP encodings have been provided by the authors \citep{BMWencodingChoiceLogics}. 

Regarding future work, new choice logics may be defined explicitly with concrete use cases in mind. Furthermore, some properties of our framework have yet to be investigated.
This includes a characterization of associativity, general concepts towards normal forms (as examined for $\qcl$ and $\ccl$ by \citet{brewka2004qualitative} and \citet{boudjelida2016conjunctive})
as well as nonmonotonic consequence relations for choice logics and how they can fit into the framework of \citet{KrausLM90} (examined for $\qcl$ by \citet{brewka2004qualitative}).
Investigating the computational complexity of these nonmonotonic consequence relations in a general manner may also yield interesting results. Furthermore, the complexity of checking for strong equivalence is still unknown for choice logics that are neither optionality-ignoring nor optionality-differentiating.

\section*{Acknowledgments}

This work was funded by the Austrian Science Fund (FWF) under the grants Y698 and P31890.

\clearpage
\newpage

\bibliographystyle{newapa}
\bibliography{references}

\end{document}